%% file: manuscript_arXiv.tex
\documentclass[aip, 
               jap, 
               reprint,
               groupedaddress
               ]{revtex4-1}

\usepackage{graphicx}
\usepackage{dcolumn}
\usepackage{bm}
\usepackage{amsmath}  
\usepackage{amssymb}   
\usepackage{amsthm}  
\usepackage{datetime}        
\usepackage{epstopdf}
\usepackage{subcaption}
\usepackage{xcolor}
\usepackage[colorlinks,linkbordercolor=white, linkcolor=blue]{hyperref}

\newtheorem{theorem}{\bf Theorem}

\newtheorem{corollary}{\bf Corollary}
\newtheorem{remark}{\bf Remark}

\input{preamble_arXiv.tex}

\begin{document}

\title{Critical review of Murray's theory for optimal branching in fluidic networks}


\author{R. Hagmeijer}
\affiliation{University of Twente, Faculty of Engineering Technology, P.O. Box 217, 7500 AE Enschede, The Netherlands}
\email{r.hagmeijer@utwente.nl}
\author{C.H. Venner}
\affiliation{University of Twente, Faculty of Engineering Technology, P.O. Box 217, 7500 AE Enschede, The Netherlands}
%
%
%

\begin{abstract}
Murray's theory of constrained minimum-power branchings is critically reviewed in a generalised framework for a range of cases: channels with arbitrary cross-section shape, laminar flows of Newtonian and non-Newtonian fluids, and low and high Reynolds-number turbulent flows of Newtonian fluids. The theory states that the sum of hydraulic and metabolic power is minimised if and only if all channels satisfy the same relation between flow rate and effective radius. This relation leads to a generalised form of Murray's law. It is shown that, satisfying Murray's law is a necessary requirement for power minimisation, but not a sufficient requirement. The generalisation of Kamiya \& Togawa's law that holds for minimum-volume branchings, also holds for minimum-power branchings. It is a necessary requirement but not a sufficient requirement for both minimum-power and minimum-volume branchings. For symmetric branchings the two generalised laws of Murray and Kamiya \& Togawa become identical. 
\end{abstract}

\maketitle 
\section{Background}
During a lecture in 1809 \cite{Young 1809}, Young discussed the resistance of arterial networks. He considered a bifurcating network with symmetric branches, with a ratio of the parent channel radius and the daughter channel radii of approximately 1.26. It is not clear what the background of this rule was.

In a pioneering paper in 1926 \cite{Murray 1926a}, Murray derived an expression for the radius of an artery such that the power consumption of the artery is minimised for given flow rate. The key idea is to consider not only the power needed to maintain the {\em flow}, i.e., the product of pressure drop and flow rate, but also to maintain the {\em fluid}, i.e., the metabolic cost of blood. Assuming Hagen-Poiseuille flow of a Newtonian fluid in an artery of circular cross section, Murray derived that the cube of the radius, $R$, is proportional to the flow rate, $Q$:
\be
R^3/Q=constant.
\label{eq: R3/Q=const}
\ee
The constant is a fluid property which means that this ratio has the same value for all tubes in the arterial network. For Hagen-Poiseuille flow it additionally means that the wall shear stress has the same value in all arteries \cite{Zamir 1977}.

In the same year, Murray \cite{Murray 1926c} considered an arterial bifurcation consisting of a parent channel (index '0') and two daughter channels (indices '1' and '2'). Employing mass conservation and assuming incompressibility, i.e.,
\be
Q_0=Q_1+Q_2,
\label{eq: mass conservation bifurcation}
\ee
he derived that when \eq{eq: R3/Q=const} is satisfied in all channels, the radii of the tubes must satisfy
\be
R_0^3 = R_1^3 + R_2^3.
\label{eq: Murrays law rule of cubes}
\ee
Murray's analysis is considered the first explanation of Young's rule, since for symmetric bifurcating branchings, \eq{eq: Murrays law rule of cubes} leads to a ratio of $2^{1/3}\approx 1.25992$. 

In 1981, Sherman \cite{Sherman 1981} referred to \eq{eq: Murrays law rule of cubes} as 'Murray's law', which has been adopted by the scientific community since then. Unfortunately, Sherman also stated that \eq{eq: R3/Q=const} and \eq{eq: Murrays law rule of cubes} '{\em are alternative expressions of Murray's law}'. Kamiya {\em et al.} \cite{Kamiya etal 1974} pointed out that this is not true, since '{\em it is evident that $R_i$ and $Q_i$ ($i=0,1,2$) satisfying \eq{eq: R3/Q=const} always satisfy \eq{eq: Murrays law rule of cubes}, while inversely, $R_i$ and $Q_i$ satisfying \eq{eq: Murrays law rule of cubes}, can not always satisfy \eq{eq: R3/Q=const}}'. For example, let 
\be
R_1^3=\frac{1}{3}R_o^3, \sp R_2^3=\frac{2}{3}R_o^3, \sp Q_1=\frac{2}{3}Q_o, \sp Q_2=\frac{1}{3}Q_o, 
\ee 
then \eq{eq: mass conservation bifurcation} and \eq{eq: Murrays law rule of cubes} are both satisfied, but \eq{eq: R3/Q=const} is not:
\be
R_o^3/Q_o \ne R_1^3/Q_1 \ne R_2^3/Q_2.
\ee 
In other words, satisfying Murray's law, \eq{eq: Murrays law rule of cubes},  is a {\em necessary} requirement for power minimisation, but not a {\em sufficient} requirement. The misconception that Murray's law is a sufficient requirement has deeply entered into the literature, see for example 
Rossitti \cite{Rossitti 1995},
Dawson {\em et al.} \cite{Dawson etal 1999}, 
Painter {\em et al.} \cite{Painter etal 2006}, 
Hughes \cite{Hughes 2015}, and 
Stephenson \& Lockerby \cite{Stephenson and Lockerby 2016}.

The objective of the present paper is to critically review Murray's theory of constrained minimum-power branchings for fully developed flows in channels with arbitrary cross-section shape: laminar flows of Newtonian and non-Newtonian fluids, and low and high Reynolds-number turbulent flows of Newtonian fluids. Power minimisation  for three of these flows has been discussed in the literature, but the low-Reynolds number turbulent flow regime of a Newtonian fluid is new in this respect. It is shown that power is minimised if and only if all channels satisfy the same relation between flow rate and effective radius. Satisfying the corresponding generalised Murray law is necessary for power minimisation but not sufficient.  Kamiya \& Togawa's law that holds for minimum-volume branchings \cite{Kamiya and Togawa 1972}, also holds for minimum-power branchings. We will show that satisfying the generalisation of this law is a necessary requirement but not a sufficient requirement for both minimum-power and minimum-volume branchings. For symmetric branchings the two generalised laws of Murray and Kamiya \& Togawa become identical.

\section{Fully developed flows}
Four different cases of fully developed flow in channels are considered:
\begin{enumerate}
\item[(a)] laminar flow of a Newtonian fluid,
\item[(b)] laminar flow of a non-Newtonian fluid,
\item[(c)] low Reynolds number turbulent flow of a Newtonian fluid, smooth channel,
\item[(d)] high Reynolds number turbulent flow of a Newtonian fluid.
\end{enumerate}
The channel cross-section shapes are arbitrary except for the fourth regime where we assume a circular cross section. In all of these cases the pressure drop $\Delta p$ over the channel, i.e., the difference between the pressure at the entrance and the pressure at the exit, can be written in terms of the Darcy-Weisbach formulation:
\be
\Delta p = f\frac{1}{2}\rho U^2 \frac{L}{2R}.
\ee
In this expression, $f$ is the friction factor, $\rho$ is the mass density, $U$ is the cross-section averaged velocity, $L$ is the length of the channel, and $R$ is the effective channel radius, defined as the radius of a circular channel with the same cross-section area $A$:
\be
R\equiv\sqrt{\frac{A}{\pi}}.
\label{eq: effective radius}
\ee
By introducing the flow rate $Q=U\pi R^2$ one gets
\be
\Delta p = f\,\frac{\rho}{4\pi^2}\frac{Q^2}{R^5}L.
\label{eq: delta p}
\ee
For a Newtonian fluid with viscosity $\mu$ and a channel with average wall roughness $e$, the friction factor $f$ is a function of two dimensionless parameters: the Reynolds number,
\be
Re \equiv \frac{2 \rho U R}{\mu} = \frac{2}{\pi}\frac{\rho Q}{\mu R},
\label{eq: Reynolds number}
\ee
and the relative wall roughness
\be
\epsilon\equiv\frac{e}{2R}.
\ee 
For a non-Newtonian fluid with the viscosity satisfying a power law of the form
\be
\mu = \mu' |\dot{\gamma}|^{n-1},  \sp n\in\real^+,
\label{eq: powerlaw}
\ee
with $\dot{\gamma}$ the shear rate and $\mu'$ a constant, the friction factor is a function of of three dimensionless parameters: the non-dimensional group
\be
Re'\equiv \frac{\rho U^{2-n} \left(2R\right)^n}{\mu'} = \frac{2^n}{\pi^{2-n}}\frac{\rho}{\mu'}\frac{Q^{2-n}}{R^{4-3n}},
\ee
the relative wall roughness $\epsilon$, and the exponent $n$. When $n=1$ one recovers the constant viscosity model of a Newtonian fluid, and $Re'$ reduces to $Re$. For fixed wall roughness $e$, fixed viscosity coefficient $\mu$ or fixed $\mu'$ and $n$, one can write \eq{eq: delta p} as:
\be
\Delta p = c \,Q^a R^{-b} L.
\label{eq: dp in channel, generic}
\ee
In this expression, $a$ and $b$ are positive constants which depend on the flow regime, and $c$ is a positive constant that depends on the flow regime and on the the cross-section shape of the channel. The values of $a$, $b$, and $c$ are derived in the next paragraphs for the flows considered, and summarised in \tab{tab: a,b,c}.

\def\arraystretch{1.5}
\begin{table*}[t]\centering
\begin{center}
\begin{tabular}{|l|c|c|l|}
\hline
Flow regime & $a$ & $b$ & $c$ \\ \hline
laminar Newtonian & 1 & 4 & circular: $\frac{8\mu}{\pi}$ \\ 
                      &   &    & elliptic: $\left(\frac{\left(h_1/h_2\right)^2 + 1}{\left(h_1/h_2\right)}\right)\frac{4\mu}{\pi}.$ \\ 
                      &   &    & square: $\frac{16\mu}{0.562\pi^2} \approx \frac{9.062\mu}{\pi}.$ \\ 
laminar non-Newtonian & $n$ & $3n+1$ & circular: $2\mu'\left(\frac{3n+1}{n\pi}\right)^n$ \\ 
low-$Re$ turbulent Newtonian & $\frac{7}{4}=1.75$ & $\frac{19}{4}=4.75$ & circular: $6.64\times 10^{-2}\,\pi^{-\frac{7}{4}}\,\,
\mu^{\frac{1}{4}}\rho^{\frac{3}{4}}$ \\ 
high-$Re$ turbulent Newtonian & 2 & 5 & 
circular: $\frac{\rho}{4\pi^2}\left\{-1.8\log_{10}\left(\frac{\epsilon}{3.7}\right)\right\}^{-1}$  
\\ \hline
\end{tabular}
\end{center}
\caption{Summary of values of $a$, $b$, and $c$ in \eq{eq: dp in channel, generic} for the flows considered.}
\label{tab: a,b,c}
\end{table*}%

\subsection{Laminar flow of Newtonian fluid}
The fully developed laminar flow of a Newtonian fluid in a branching of smooth channels ($\epsilon = 0$) of arbitrary cross-section was considered by Emerson {\em et al.} \cite{Emerson etal 2006}. The axial velocity $w$ satisfies 
\be
\pardiff{w}{x} + \pardiff{w}{y} = \frac{1}{\mu}\dif{p}{z},
\ee
where $x$ and $y$ are the cartesian coordinates in the cross-sectional plane and $z$ is the cartesian coordinate along the channel. The differential equation shows that $w\sim\frac{1}{\mu}\dif{p}{z}$. For a given cross-section shape and effective radius $R$, the resulting flow rate $Q$ is a function of $\frac{1}{\mu}\dif{p}{z}$ and $R$. Dimension analysis  leads to
\be
Q \sim \frac{1}{\mu}\dif{p}{z} R^4 \sp 
\Rightarrow \sp \dif{p}{z}\sim \frac{\mu Q}{R^4},
\ee
and therefore
\be
a=1, \sp b=4.
\ee 

Several examples belonging to this class of flows are known. In case of a circular channel with Hagen-Poiseuille flow, the flow rate is given by Lamb \cite{Lamb 1932},
\be
Q = -\frac{\pi}{8}\frac{1}{\mu}\dif{p}{z} R^4, 
\ee
such that
\be
f=\frac{64}{Re},  \sp 
c=\frac{8\mu}{\pi}.
\ee
For an elliptic channel with semi-axes $h_1$ and $h_2$, the effective radius is $R=\sqrt{h_1h_2}$, and the flow rate is again given by Lamb \cite{Lamb 1932}:
\be
Q = -\frac{\pi}{4}\frac{h_1^3 h_2^3}{h_1^2 + h_2^2}\frac{1}{\mu}\dif{p}{z} = -\frac{\pi}{4}\frac{h_1 h_2}{h_1^2 + h_2^2}\frac{1}{\mu}\dif{p}{z}R^4.
\ee
As a consequence, 
\be
f=\left(\frac{\left(h_1/h_2\right)^2 + 1}{\left(h_1/h_2\right)}\right)\frac{32}{Re},  \sp 
c=\left(\frac{\left(h_1/h_2\right)^2 + 1}{\left(h_1/h_2\right)}\right)\frac{4\mu}{\pi}.
\ee
Finally, for a square channel with sides $2h$, the effective radius is $R=\frac{2h}{\sqrt{\pi}}$, and the flow rate is given by Cornish \cite{Cornish 1928}:
\be
\begin{split}
Q &= -\frac{4}{5}\frac{h^4}{\mu}\dif{p}{z}
\left(
1 - \frac{192}{\pi^5}\sum_{n=0}^\infty \frac{ \tanh{\left((2n+1)\frac{\pi}{2}\right)}    }{    (2n+1)^5    }
\right)\\
&\approx -0.562\frac{h^4}{\mu}\dif{p}{z} = -0.562\frac{\pi^2}{16}\frac{1}{\mu}\dif{p}{z} R^4.
\end{split}
\ee
\be
f\approx\frac{128}{0.562\pi Re} \approx \frac{72.50}{Re}, \sp 
c\approx\frac{16\mu}{0.562\pi^2} \approx \frac{9.062\mu}{\pi}.  
\ee

\subsection{Laminar flow of non-Newtonian fluid}
The fully developed laminar flow of a non-Newtonian fluid in a branching of smooth channels ($\epsilon = 0$) of arbitrary cross-section was considered by Revellin {\em et al.} \cite{Revellin etal 2009} and by Tesch \cite{Tesch 2010}. The axial velocity $w$ satisfies 
\be
\pardif{}{x}\left(\mu\pardif{w}{x}\right) + \pardif{}{y}\left(\mu\pardif{w}{y}\right) = \frac{1}{\mu}\dif{p}{z},
\label{eq: differential equation nN flow}
\ee
where $x$ and $y$ are the cartesian coordinates in the cross-sectional plane, $z$ is the cartesian coordinate along the channel, and $\mu$ is given by \eq{eq: powerlaw}. The shear rate is defined as \cite{Bird Stewart Lightfoot}:
\be
\dot{\gamma} \equiv \sum_{i=1}^3\sum_{j=1}^3\sqrt{\frac{1}{2}\gamma_{ij}\gamma_{ij}}, \sp
\gamma_{ij}\equiv \pardif{u_i}{x_j} + \pardif{u_j}{x_i},
\ee
which in the present case leads to
\be
\dot{\gamma} = \sqrt{\left(\pardif{w}{x}\right)^2 + \left(\pardif{w}{y}\right)^2}.
\ee
The differential equation \eq{eq: differential equation nN flow} and the power law \eq{eq: powerlaw} show that $w^n\sim\frac{1}{\mu}\dif{p}{z}$. For a given cross-section shape and effective radius $R$, the resulting flow rate $Q$ is a function of $\frac{1}{\mu}\dif{p}{z}$ and $R$, and dimension analysis leads to
\be
Q^n \sim \frac{1}{\mu'}\dif{p}{z} R^{3n+1} \sp 
\Rightarrow \sp \dif{p}{z}\sim \frac{\mu' Q^n}{R^{3n+1}},
\ee
and therefore 
\be
a=n, \sp b=3n+1.
\ee

In case of a circular channel the relation between flow rate and pressure gradient is given by Bird {\em et al.} \cite{Bird Stewart Lightfoot}:
\be
Q = \frac{n\pi R^3}{3n+1}\left(-\dif{p}{z}\frac{R}{2\mu'}\right)^{1/n},
\ee
and therefore
\be
f = 2^{n+3}\left(\frac{3n+1}{n}\right)^n \frac{1}{Re'}, \sp
c = 2\mu'\left(\frac{3n+1}{n\pi}\right)^n.
\ee

\subsection{Low Reynolds number turbulent flow of Newtonian fluid, smooth channel}
When the flow is turbulent and the Reynolds number is sufficiently low, $Re<10^5$, then the friction factor for a smooth channel, $\eps=0$, may be approximated by Blasius' formula \cite{Blasius 1912, Blasius 1913}:
\be
f=\frac{0.3164}{Re^{\frac{1}{4}}},
\label{eq: Blasius formula}
\ee
and the coefficients in \eq{eq: dp in channel, generic} for this flow regime become 
\be
a=\frac{7}{4},  \sp 
b=\frac{19}{4}, \sp 
c=6.64\times 10^{-2}\,\pi^{-\frac{7}{4}}\,\,\mu^{\frac{1}{4}}\rho^{\frac{3}{4}}.
\ee

\subsection{High Reynolds number turbulent flow of Newtonian fluid}
For sufficiently large Reynolds numbers, the friction factor corresponding to fully developed turbulent flow in a circular channel can quite accurately be described by Haaland's formula \cite{Haaland 1983}, which in the limit of high Reynolds numbers becomes:
\be
f = \left\{-1.8\log_{10}\left(\frac{\epsilon}{3.7}\right)\right\}^{-1}.
\label{eq: Haaland friction factor, high Re}
\ee
The coefficients in \eq{eq: dp in channel, generic} for this flow regime become 
\be
a=2, \sp
b=5, \sp
c=\frac{\rho}{4\pi^2}\left\{-1.8\log_{10}\left(\frac{\epsilon}{3.7}\right)\right\}^{-1}.
\ee

\section{Generalisation of Murray's theory for a single channel}
Murray's conjecture \cite{Murray 1926a} is that, at fixed channel length $L$ and flow rate $Q$, the channel radius-dependent power $P(R)$ consists of two contributions: one to maintain the flow rate against an adverse pressure gradient $\Delta p$, and one to maintain the fluid:
\be
P(R) \equiv \Delta p \, Q + \alpha V.
\label{eq: channel power}
\ee
In this expression, $V$ is the channel volume,
\be
V=\pi R^2 L,
\label{eq: channel volume}
\ee
and $\alpha$ is a fluid maintenance constant representing the cost per unit volume to maintain the fluid. Murray minimised the power with respect to $R$ assuming Hagen-Poiseuille flow and found that the optimal radius $R_*$ is proportional to the cube root of the flow rate. Furthermore, the corresponding power required to maintain the flow rate was found to be $\frac{1}{2}\alpha V_*$ with $V_*=\pi R_*^2 L$. The minimum power to maintain both the flow rate and the fluid is $P_*=\frac{3}{2}\alpha V_*$, which leads to the conclusion that the ratio of the power required to maintain the flow rate to the power required to maintain the fluid is $\frac{1}{2}$. Finally, Uylings \cite{Uylings 1977} derived an expression for the ratio of the non-optimised power and the power minimum:
\be
\frac{P}{P_*} = \frac{1}{3}\left(\frac{R}{R_*}\right)^{-4} + \frac{2}{3}\left(\frac{R}{R_*}\right)^{2}.
\ee
All of these results obtained for Hagen-Poiseuille flow can be generalised towards the generalised pressure-drop flow-rate relation \eq{eq: dp in channel, generic}, for channels with the effective radius defined in \eq{eq: effective radius}. The generalisation is given by the following theorem.

\begin{theorem}[minimum-power channel]
\label{th: channel}
The power $P(R)$ required to maintain a fully developed steady flow at fixed flow rate $Q$ in a channel of fixed length $L$ with effective radius $R$, pressure drop $\Delta p = c \,Q^a R^{-b} L$ and fluid maintenance coefficient $\alpha$, attains a global minimum if and only if $R=R_*$ with
\be
R_*^\frac{b+2}{a+1}
\equiv \left(\frac{bc}{2\pi\alpha}\right)^{1/(a+1)} Q.
\label{eq: R channel}
\ee
The global minimum of $P$ is
\be 
P_*=\left(\frac{b+2}{b}\right)\alpha V_*,
\label{eq: power minimum channel}
\ee
and the corresponding ratio of the two power contributions is
\be
\left(\frac{\Delta p \,Q}{\alpha V}\right)_*=\frac{2}{b}.
\label{eq: power ratio channel}
\ee
Finally, the ratio of the power $P$ and its minimum value $P_*$ is:
\be
\frac{P}{P_*} = \frac{2}{b+2}\left(\frac{R}{R_*}\right)^{-b} + \frac{b}{b+2}\left(\frac{R}{R_*}\right)^{2}.
\label{eq: dimensionless power}
\ee
\end{theorem}
\noindent
\hrulefill
\begin{proof} 
$P$ is a function of $R$ only and 
\be
\dif{P}{R}=\left(-bc\frac{Q^{a+1}}{R^{b+2}} + 2\alpha\pi\right)RL,
\ee
which shows that $\dif{P}{R}=0$ if and only if \eq{eq: R channel} holds. Furthermore
\be
\diff{P}{R}=\left(b\,(b+1)c\frac{Q^{a+1}}{R^{b+2}} + 2\alpha\pi\right)L,
\ee
which is positive for all $R$ showing that the minimum is a global minimum. The two expressions for the minimum power and the power ratio follow immediately by substitution. Finally, when one divides $P(R)$ by $P_*$ using \eq{eq: power minimum channel}, and by using \eq{eq: R channel} to substitute
\be
Q^{a+1}=\frac{2\alpha\pi}{bc} R_*^{b+2},
\ee
then \eq{eq: dimensionless power} follows immediately.
\end{proof}

\noindent The characteristic numbers appearing in \thm{th: channel} are summarised in \tab{tab: characteristic numbers} for the flows considered.
%
\def\arraystretch{1.5}
\begin{table*}[t]\centering
\begin{center}
\begin{tabular}{|l|c|c|c|c|c|c|c|}
\hline
Flow regime & $a$ & $b$ & $\frac{b+2}{a+1}$ & $\frac{b+2}{b}$ & $\frac{2}{b}$ & $\frac{2}{b+2}$ & $\frac{b}{b+2}$ \\ \hline
laminar Newtonian & 1 & 4 & $3$ & $\frac{3}{2}$ & $\frac{1}{2}$ & $\frac{1}{3}$ & $\frac{2}{3}$        \\
laminar non-Newtonian & $n$ & $3n+1$ & $3$ & $\frac{3n+3}{3n+1}$ & $\frac{2}{3n+1}$ & $\frac{2}{3n+3}$ & $\frac{3n+1}{3n+3}$ \\
low-$Re$ turbulent Newtonian & $\frac{7}{4}$ & $\frac{19}{4}$ & $\frac{27}{11}\approx2.45$ & $\frac{27}{19}$ & $\frac{8}{19}$ & $\frac{8}{27}$ & $\frac{19}{27}$ \\
high-$Re$ turbulent Newtonian & 2 & 5 & $\frac{7}{3}\approx2.33$ & $\frac{7}{5}$ & $\frac{2}{5}$ & $\frac{2}{7}$ & $\frac{5}{7}$ \\ \hline
\end{tabular}
\end{center}
\caption{Summary of characteristic numbers for the flows considered.}
\label{tab: characteristic numbers}
\end{table*}%

\section{Generalisation of Murray's theory for a branching}
Following the single-channel result in the previous section, we now consider a branching consisting of a parent channel connected to $N$ daughter channels in a branching point $\vx$, see \fig{fig: branching}. The channels are numbered from $0$ to $N$, with $0$ indicating the parent channel. The effective radii of the channels are $\vR\equiv(R_0, R_1,...,R_N)$, the fixed termination points of the channels are $\vx_i$, $i=0,1,...,N$, and the fixed flow rates in the daughter channels are $Q_i$, $i=1,...,N$. Furthermore, $Q_o$ is taken positive towards the branching point, whereas the other flow rates are taken positive away from the branching point. To satisfy mass conservation, the flow rates satisfy:
\be
Q_o=\sum_{i=1}^N Q_i.
\label{eq: mass conservation}
\ee
Finally, the lengths of the channels, $L_i$, are functions of the branching location:
\be
L_i\equiv|\vx_i-\vx|, \sp i=0,1,...,N.
\ee

\begin{figure}[h!]
\begin{center}
   \centering
   \includegraphics[width=0.4\textwidth]{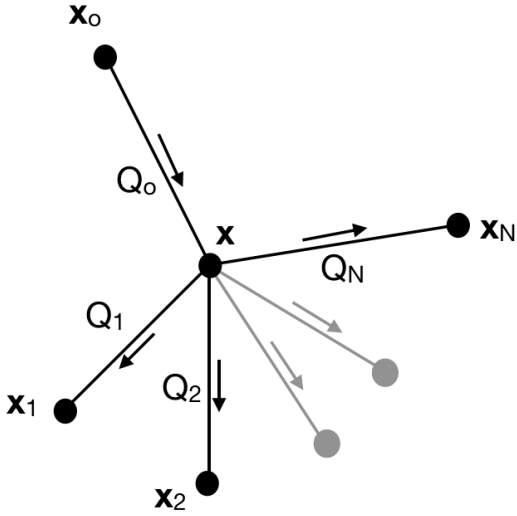} 
\caption{Branching with parent channel and $N$ daughter channels}
\label{fig: branching}
\end{center}
\end{figure}

\subsection{Minimised-power branching}
\label{se: Minimised-power branching}
The power $P(\vR,\vx)$ needed to maintain the flow rate and the fluid in the channel depends on the radii and lengths of the channels, and is the sum of the individual channel contributions given by \eq{eq: channel power}:
\be
P(\vR,\vx) \equiv \sum_{i=0}^N \left\{\Delta p \, Q + \alpha V\right\}_i,
\label{eq: branching power}
\ee
Assuming Hagen-Poiseuille flow in cylindrical channels, Murray \cite{Murray 1926c} derived for such a branching that, based on mass conservation and assuming optimised channels such that the radii are proportional to the cube roots of the flow rates, the sum of the cubes of the daughter radii must be equal to the cube of the parent radius. In the literature this relation is referred to as 'Murray's law'. Furthermore, Murray was able to derive expressions for the cosines of the angles between the channels in the bifurcation case ($N=2$), pre-assuming that all channels lie in a plane. The following theorem generalises these results.

\begin{theorem}[minimum-power branching]
\label{th: Murray, branching}
The power $P(\vR,\vx)$ required to maintain fully developed steady flows at fixed flow rates $Q_i$ in the branching channels of fixed lengths $L_i$ with effective radii $R_i$, pressure drops $\Delta p_i = c \,Q_i^a R_i^{-b} L_i$ and fluid maintenance coefficient $\alpha$, attains a global minimum if and only if $R_i=R_{i,*}$ with
\be
R_{i,*}^\frac{b+2}{a+1}\equiv \left(\frac{bc}{2\pi\alpha}\right)^{1/(a+1)} Q_i, \sp i=0,1,...,N,
\label{eq: R branching}
\ee
and $\vx=\vx_*$ with
\be
\sum_{i=0}^N R_{i,*}^2 \ve_{i,*}=0, \sp 
\ve_{i,*}\equiv\left(\vnabla L_i\right)_*=\frac{\vx_*-\vx_i}{|\vx_*-\vx_i|}.
\label{eq: branching point equations}
\ee
The global minimum of $P$ is
\be 
P_*=\left(\frac{b+2}{b}\right)\alpha \sum_{i=0}^N V_{i,*},
\label{eq: power minimum branching}
\ee
and the corresponding ratio of the two power contributions is
\be
\left(\frac{\sum_{i=0}^N\Delta p_i \,Q_i}{\sum_{i=0}^N\alpha V_i}\right)_*=\frac{2}{b}.
\label{eq: power ratio branching}
\ee
Finally, the following relation holds:
\be
R_{o,*}^\frac{b+2}{a+1}\left(\frac{R_{o,*}^\frac{b+2}{a+1}}{Q_o}\right)^m=
\sum_{i=1}^N R_{i,*}^\frac{b+2}{a+1}\left(\frac{R_{i,*}^\frac{b+2}{a+1}}{Q_i}\right)^m
,\,\,\, \forall m\in\real.
\label{eq: extended Murray's law}
\ee
\end{theorem}
\noindent
\hrulefill
\begin{proof} 
Differentiation of $P$ with respect to $R_i$ gives
\be
\pardif{P}{R_i}=\left(-bc\frac{Q_i^{a+1}}{R_i^{b+2}} + 2\alpha\pi\right)R_i L_i,
\ee
which shows that $\pardif{P}{R_i}=0$ if and only if \eq{eq: R branching} holds.

Furthermore, the gradient of $P$ with respect to the branching point $\vx$ is:
\be
\nabla_{\vx}{P}=\left(c\frac{Q_i^{a+1}}{R_i^{b}} + \alpha\pi R_i^2\right) \vnabla L_i.
\label{eq: nabla P}
\ee
Because $L^2_i=|\vx-\vx_i|^2=\left(\vx-\vx_i\right)\cdot\left(\vx-\vx_i\right)$ we have
\be
2 L_i \vnabla L_i = \vnabla L^2_i = 2\left(\vx - \vx_i\right),
\ee
and therefore
\be
\vnabla L_i = \ve_i \equiv \frac{\vx-\vx_i}{|\vx-\vx_i|}.
\ee
By using \eq{eq: R branching} to eliminate $Q$ from \eq{eq: nabla P}, one obtains that $\nabla_{\vx}{P}=0$ if and only if
\be
\frac{b+2}{b}\alpha \pi \sum_{i=0}^N R_{i,*}^2 \ve_{i,*}=0.
\ee
Since $b>0$ and $\alpha>0$, this immediately implies \eq{eq: branching point equations}. 

\eq{eq: power minimum branching} can be found by substitution of $R_{i,*}$ and $\vx_*$ into the expression for $P(\vR,\vx)$.

To show that the power minimum $P_*$ is indeed a global minimum we write $P$ as a sum over the individual channel contributions (see \eq{eq: branching power}, \eq{eq: dimensionless power} and \eq{eq: power minimum branching}):
\be
\begin{split}
P&=\sum_{i=0}^N P_i = \sum_{i=0}^N \frac{P_i}{P_{i,*}}P_{i,*} \\
&= \sum_{i=0}^N \left\{ \frac{2}{b+2}\left(\frac{R_i}{R_{i,*}}\right)^{-b} + \frac{b}{b+2}\left(\frac{R_i}{R_{i,*}}\right)^{2} \right\}_i \\
&\times 
\left(\frac{b+2}{b}\right)\alpha\pi  R_{i,*}^2 L_{i},
\end{split}
\label{eq: P as a sum}
\ee
It should be noted that $P_{i,*}$ in this expression denotes the minimum power of channel $i$ for given length $L_i$, i.e., it has only be optimised with respect to $R_i$. The terms in between brackets in \eq{eq: P as a sum} are either larger than one,  or equal to one if and only if $R_i/R_{i,*}=1$. This can be seen from considering the function $f(x)\equiv \frac{2}{b+2} x^{-b} + \frac{b}{b+2} x^{2}$ for $x>0$ which has global minimum $f(1)=1$ since $f'(1)=0$ and $f''(x)>0$. Therefore
\be
P\left(\vR,\vx\right) \ge \left(\frac{b+2}{b}\right)\alpha\pi \sum_{i=0}^N  R_{i,*}^2 L_{i},
\label{eq: P as a simplified sum}
\ee
with equality if and only if $\frac{R_i}{R_{i,*}}=1$ for all $i$. It remains to be shown that the sum in \eq{eq: P as a simplified sum} has a global minimum when the branching point $\vx$ satisfies \eq{eq: branching point equations}. We write the branching point as a perturbation of the optimum:
\be
\vx = \vx_* + s\vr, \sp s\in \real, \sp \vr\in\real^3, \sp |\vr|=1.
\ee
A Taylor series expansion shows that 
\be
\begin{split}
\sum_{i=0}^N  R_{i,*}^2 L_{i} &= \left(\sum_{i=0}^N  R_{i,*}^2 L_i \right)_{s=0} \\
&+\left(\dif{}{s}\sum_{i=0}^N  R_{i,*}^2 L_i \right)_{s=0}s \\
&+\int_0^s\int_0^t\left(\diff{}{s}\sum_{i=0}^N  R_{i,*}^2 L_i \right)_{s=u}\,du\,dt,
\end{split}
\label{eq: expansion of sum}
\ee
The first and second derivatives in this expression are, respectively:
\be
\dif{}{s}\sum_{i=0}^N  R_{i,*}^2 L_i = \vnabla \left(\sum_{i=0}^N  R_{i,*}^2 L_i\right) \cdot \dif{\vx}{s} =
\left(\sum_{i=0}^N  R_{i,*}^2 \ve_i\right)\cdot\vr,
\label{eq: 1st derivative to s}
\ee
and
\be
\diff{}{s}\sum_{i=0}^N  R_{i,*}^2 L_i = \sum_{i=0}^N  R_{i,*}^2 \dif{\ve_i}{s}\cdot\vr=
\sum_{i=0}^N  \frac{R_{i,*}^2}{L_i} \left\{1-\left(\ve_i\cdot\vr\right)^2\right\},
\label{eq: 2nd derivative to s}
\ee
where we have used
\be
\dif{\ve_i}{s}=\frac{1}{L_i} \left\{\vr - \ve_i \left(\vnabla L_i \cdot \dif{\vx}{s}\right)\right\} =
\frac{1}{L_i} \left\{\vr - \ve_i \left(\ve_i\cdot\vr\right)\right\}.
\ee
With these expressions, \eq{eq: expansion of sum} can be written as
\be
\begin{split}
\sum_{i=0}^N  R_{i,*}^2 L_{i} &= \left(\sum_{i=0}^N  R_{i,*}^2 L_i \right)_{s=0}
+ \left(\sum_{i=0}^N  R_{i,*}^2 \ve_{i,*}\right)\cdot\vr s \\
&+ \int_0^s\int_0^t\left(\sum_{i=0}^N  \frac{R_{i,*}^2}{L_i} \left\{1-\left(\ve_i\cdot\vr\right)^2\right\}\right)_{s=u}\,du\,dt.
\end{split}
\ee
The second term on the right hand side is zero in view of \eq{eq: branching point equations}, and the third term on the right hand side of \eq{eq: expansion of sum} is non-negative since $|\ve_i|=1$, $|\vr|=1$ and therefore $\left(\ve_i\cdot\vr\right)^2\le 1$ with the inequality applying to at least one of the channels. Hence
\be
\sum_{i=0}^N  R_{i,*}^2 L_{i} \ge \left(\sum_{i=0}^N  R_{i,*}^2 L_i \right)_{s=0},
\ee
and therefore the power minimum is a global minimum.

\eq{eq: power ratio branching} follows directly from substitution of \eq{eq: R branching} into the expressions for $\Delta p_i$ and $V_i$ given by \eq{eq: dp in channel, generic} and \eq{eq: channel volume}, respectively.

Finally we prove \eq{eq: extended Murray's law} first by replacing the flow rates $Q_i$ in the mass conservation law \eq{eq: mass conservation} by means of \eq{eq: R branching}:
\be
R_{o,*}^\frac{b+2}{a+1}=\sum_{i=1}^N R_{i,*}^\frac{b+2}{a+1},
\ee
and then by multiplying each term $R_{i,*}^\frac{b+2}{a+1}$ by the factor $\left(R_{i,*}^\frac{b+2}{a+1}/Q_i\right)^{m}$,
which is independent of $i$ in view of \eq{eq: R branching}, for arbitrary values of $m$.
\end{proof}

\begin{remark}
If \eq{eq: R branching} is satisfied, then all channels are optimised in the sense that the powers corresponding to the  channels are minimised individually for fixed lengths $L_i$. In contrast, if \eq{eq: branching point equations} is satisfied, then the total volume of the branching is minimised for fixed radii $R_i$.
\end{remark}

\begin{corollary}[bifurcation topology]
The optimal branching point $\vx_*$ of a bifurcation, i.e., $N=2$, lies in the plane defined by $\vx_o$, $\vx_1$,  $\vx_2$, and the cosines of the smallest angles between each pair of channels involved are given by
\be
\left.\begin{array}{crr}
\ve_o\cdot\ve_1&=&-\frac{R_o^4+R_1^4-R_2^4}{2R_o^2R_1^2}, \\
\ve_o\cdot\ve_2&=&-\frac{R_o^4-R_1^4+R_2^4}{2R_o^2R_2^2}, \\
\ve_1\cdot\ve_2&=&\frac{R_o^4-R_1^4-R_2^4}{2R_1^2R_2^2}.
\label{eq: angle cosines branching}
\end{array}\right.
\ee
\end{corollary}
\begin{proof} 
\eq{eq: branching point equations} shows that the vectors $\ve_i$, $i=0,1,2$, are linearly dependent, which means they lie in the same plane and, as a consequence, the optimal branching point $\vx_*$ lies in the plane defined by $\vx_o$, $\vx_1$,  $\vx_2$. Taking the inner product of \eq{eq: branching point equations} with the vectors $\ve_o$, $\ve_1$, and $\ve_2$, respectively, leads to the following linear system of equations:
\be
\matrixthree{R_1^2}{R_2^2}{0}
{R_o^2}{0}{R_1^2}
{0}{R_o^2}{R_1^2}
\vectorthree{\ve_o\cdot\ve_1}{\ve_o\cdot\ve_2}{\ve_1\cdot\ve_2}=
-\vectorthree{R_o^2}{R_1^2}{R_2^2},
\ee
which has unique solution \eq{eq: angle cosines branching}.
\end{proof}

\begin{remark}
\eq{eq: R branching} implies \eq{eq: extended Murray's law}, but \eq{eq: extended Murray's law} does not imply \eq{eq: R branching}. Instead, \eq{eq: extended Murray's law} implies
\be
R_{i,*}^\frac{b+2}{a+1}\equiv C Q_i, \sp i=1,2,...,N,
\ee
with the constant $C$ undetermined, and therefore \eq{eq: extended Murray's law} is not a weak formulation of \eq{eq: R branching}.
\end{remark}

\subsection{Generalised Murray and Kamiya-Togawa laws}

\paragraph{Murray's law.} For $m=0$, and leaving out the asterisks, \eq{eq: extended Murray's law} becomes
\be
R_{o}^\frac{b+2}{a+1}=\sum_{i=1}^N R_{i}^\frac{b+2}{a+1},
\label{eq: generalised Murray law N}
\ee
which we will refer to as the {\bf generalised Murray law}. \eq{eq: generalised Murray law N} defines a hyper-surface in the $(N+1)$-dimensional space of radii $R_{i}$. In contrast, \eq{eq: R branching} defines a single point on that hyper-surface. For that reason, both equations are not equivalent: \eq{eq: generalised Murray law N} is only a {\em necessary} condition for power-minimisation, whereas \eq{eq: R branching} is a {\em sufficient} condition for power-minimisation. In the special case of a bifurcation, $N=2$, and \eq{eq: generalised Murray law N} reduces to
\be
R_{o}^\frac{b+2}{a+1}=R_{1}^\frac{b+2}{a+1} + R_{2}^\frac{b+2}{a+1}.
\label{eq: generalised Murray law N=2}
\ee
For Hagen-Poiseuille flow this expression further reduces to the original law \eq{eq: Murrays law rule of cubes}:
\begin{equation}
R_{o}^3=R_{1}^3 + R_{2}^3.
\nonumber
\end{equation}

\paragraph{Kamiya-Togawa's law.} For $m=a$, leaving out the asterisks, \eq{eq: extended Murray's law} becomes
\be
\frac{R_{o}^{b+2}}{Q_o^a}=
\sum_{i=1}^N \frac{R_{i}^{b+2}}{Q_i^a},
\label{eq: generalised K&T law N}
\ee
which we will refer to as the {\bf generalised Kamiya-Togawa law}. \eq{eq: generalised K&T law N} defines a hyper-surface in the $(N+1)$-dimensional space of radii $R_{i}$, and \eq{eq: R branching} defines a single point of that hyper-surface. For $N=2$, \eq{eq: generalised K&T law N}  reduces to
\be
\frac{R_{o}^{b+2}}{Q_o^a}=
\frac{R_{1}^{b+2}}{Q_1^a} + \frac{R_{2}^{b+2}}{Q_1^a},
\label{eq: generalised K&T law N = 2}
\ee
and for Hagen-Poiseuille flow ($a=1$, $b=4$), this expression further reduces to 
\be
\frac{R_{o}^6}{Q_o}=
\frac{R_{1}^6}{Q_1}+\frac{R_{2}^6}{Q_2}.
\label{eq: K&T law}
\ee
This equation was derived by Kamiya \& Togawa \cite{Kamiya and Togawa 1972} as the result of volume minimisation for fixed flow rates and fixed pressure drops between the branching-entrance and exits. It is easy to show that \eq{eq: generalised K&T law N} similarly follows from volume minimisation for the generalised case. For fixed flow rates and fixed pressure drops between the branching-entrance and exits, i.e.,
\be
c Q_o^a R_o^{-b} L_o + c Q_i^a R_i^{-b} L_i = \mbox{constant} \sp \forall i>0,
\ee
the radii $R_i$ for $i>0$ become functions of the radius $R_o$ and the branching point $\vx$. Differentiation of this expression to $R_o$ gives:
\be
\pardif{R_i}{R_o}=-\frac{L_o}{L_i}\left(\frac{R_o}{R_i}\right)^{-(b+1)}\left(\frac{Q_o}{Q_i}\right)^a.
\label{eq: dRidRo}
\ee
Minimisation of the total branching volume $V=\sum_{i=0}^N \pi R_i^2 L_i$ requires $\pardif{V}{R_o}=0$ which, together with \eq{eq: dRidRo}, leads to the generalised law\eq{eq: generalised K&T law N}. Hence, \eq{eq: generalised K&T law N} apparently is a consequence of {\em power}-minimisation and a consequence of {\em volume}-minimisation. It therefore represents a necessary condition for both types of minimisation.

\paragraph{Symmetry.} \eq{eq: generalised Murray law N}, which is a necessary condition for power-minimisation, and \eq{eq: generalised K&T law N}, which is a necessary condition for {\em both} volume-minimisation {\em and} power-minimisation, are in general not equivalent since they define two different hyper-surfaces. However, in the special case of a symmetric branching,
\be
R_i=R_1, \sp Q_i=Q_1, \sp i=2,3,...,N,
\ee
\eq{eq: generalised Murray law N} and \eq{eq: generalised K&T law N} become 
\be
\frac{R_i}{R_o}=N^{-\frac{a+1}{b+2}}, \sp\mbox{and}\sp \frac{R_i}{R_o}=N^{-\frac{1}{b+2}}\left(\frac{Q_o}{Q_i}\right)^{-\frac{a}{b+2}}, 
\ee
respectively, with $i>0$. The ratio $Q_o/Q_i$ is equal to $N$, and therefore both equations are identical.  It is noted, however, that the two corresponding branchings do not need to be identical since $R_o$ can still be different.

\subsection{Wall shear stress}
For Hagen-Poiseuille flow of a Newtonian fluid through circular tubes, power minimisation of a branching leads to uniform shear stress in all channels \cite{Zamir 1977}. We will show that this can be generalised towards laminar flows of Newtonian and non-Newtonian fluids through channels of arbitrary cross-section, but not to turbulent flows. 

The wall shear stress $\tau$ for fully developed flow through a channel with arbitrary cross section can be computed from a force balance:
\be
\Delta p A = L\oint \tau ds, 
\ee
where the closed curve integral indicates integration over the intersection between the channel wall and a perpendicular cross-plane. The average shear stress is defined as
\be
\av{\tau}\equiv \frac{1}{\ell} \oint \tau ds, \sp \ell \equiv\oint ds,
\ee
where $\ell$ is the perimeter. Hence, using \eq{eq: effective radius} and \eq{eq: dp in channel, generic}, one gets
\be
\av{\tau} = \pi c Q^a R^{2-b} / \ell.
\ee
For a fixed cross-section shape, the perimeter is linear in the effective radius, $\ell \sim R$, and therefore the average shear stress is uniform when
\be
R^\frac{b-1}{a}/Q = constant.
\label{eq: uniform wall shear stress}
\ee
For power minimisation it is required that \eq{eq: R branching} holds, so both requirements are satisfied if
\be
\frac{b-1}{a} = \frac{b+2}{a+1},
\label{eq: power minimisation and uniform wss}
\ee
or, equivalently,
\be
b = 3a + 1.
\ee
This requirement is satisfied for the laminar flows mentioned in \tab{tab: a,b,c}, but not for the turbulent flows. In \tab{tab: characteristic numbers 2} the values of $\frac{b-1}{a}$ are compared to the values of $\frac{b+2}{a+1}$ appearing in the generalised Murray law \eq{eq: generalised Murray law N}.
%
\def\arraystretch{1.5}
\begin{table}[htp]
\begin{center}
\begin{tabular}{|l|c|c|c|c|c|c|c|}
\hline
Flow regime & $a$ & $b$ & $\frac{b+2}{a+1}$ & $\frac{b-1}{a}$ \\ \hline
laminar Newtonian & 1 & 4 & $3$ & $3$ \\
laminar non-Newtonian & $n$ & $3n+1$ & $3$ & $3$  \\
low-$Re$ turbulent Newtonian & $\frac{7}{4}$ & $\frac{19}{4}$ & $\frac{27}{11}\approx2.45$ & $\frac{15}{7}\approx2.14$  \\
high-$Re$ turbulent Newtonian & 2 & 5 & $\frac{7}{3}\approx2.33$ & 2  \\ \hline
\end{tabular}
\end{center}
\caption{Summary of powers appearing in \eq{eq: power minimisation and uniform wss} for the flows considered.}
\label{tab: characteristic numbers 2}
\end{table}%

\section{Conclusions}
Murray's theory of minimum-power branchings was derived almost a century ago for channels with circular cross-section shape and Hagen-Poiseuille flow of a Newtonian fluid. It can be extended towards a range of other fully developed flows including channels with arbitrary cross-section shape, low and high Reynolds-number turbulent flows of Newtonian fluids, and laminar flows of non-Newtonian fluids. Minimisation of power is equivalent to the radii and flow rates of the branching channels satisfying the same law:
\be
\begin{split}
&\mbox{minimum power (fixed $Q_j, j=1,2,...$)} \\
&\sp\sp\sp\Leftrightarrow \\
& R_i^n/Q_i=\mbox{constant} \sp i=0,1,2,...,
\label{eq: equivalence minimum power}
\end{split}
\ee
where $n=(b+2)/(a+1)$, with $a$ and $b$ dependent on the flow regime at hand.
Taking into account mass conservation, i.e., $\sum_{i=0}^N Q_i=Q_o$, leads to a generalisation of Murray's law:
\be
\begin{split}
&\mbox{minimum power (fixed $Q_j, j=1,2,...$)}  \\
&\sp\sp\sp\Rightarrow \\
&R_{o}^{n}=\sum_{i=1}^N R_{i}^{n}.
\label{eq: consequence minimum power}
\end{split}
\ee
It is emphasised that the first statement, \eq{eq: equivalence minimum power}, expresses an {\em equivalence}, and that the second statement, \eq{eq: consequence minimum power}, expresses a {\em consequence}. In other words, satisfying Murray's law is a {\em necessary} requirement for power minimisation, but not a {\em sufficient} requirement for power minimisation,

Kamiya \& Togawa's theory of minimum-volume branchings can also be extended to the flows mentioned above:
\be
\begin{split}
&\mbox{minimum volume (fixed $Q_j, \Delta p_{oj}, j=1,2,...$)} \\
&\sp\sp\sp\Rightarrow \\
&\frac{R_{o}^{b+2}}{Q_o^a}=\sum_{i=1}^N \frac{R_{i}^{b+2}}{Q_i^a}.
\label{eq: consequence minimum volume}
\end{split}
\ee
Minimum-{\em power} branchings also satisfy this generalised law:
\be
\begin{split}
&\mbox{minimum power (fixed $Q_j, j=1,2,...$)} \\
&\sp\sp\sp\Rightarrow \\
& \frac{R_{o}^{b+2}}{Q_o^a}=\sum_{i=1}^N \frac{R_{i}^{b+2}}{Q_i^a},
\label{eq: consequence minimum power, 2}
\end{split}
\ee
which reflects that \eq{eq: consequence minimum volume} and \eq{eq: consequence minimum power, 2} both represent {\em consequences} and not {\em equivalences}. 

For a symmetric branching, the two consequences \eq{eq: consequence minimum power} and \eq{eq: consequence minimum volume} become identical:
\be
\frac{R_i}{R_o}=N^{-1/n}, \sp i=1,2,...,
\ee
but the branchings can still be different because $R_o$ can be different.

Finally it has been shown that the requirements of power minimisation on the one hand, and uniform perimeter-averaged wall shear stress on the other, both lead to the same Murray law in case of laminar flow of Newtonian and non-Newtonian fluids, but to different laws in case of turbulent flow.


\end{document}

%% file: preamble_arXiv.tex
\usepackage[british,UKenglish,USenglish,english,american]{babel}

\def\sp{\;\;\;\;}

\def\pardif#1#2{ \frac{\partial #1 }{ \partial #2 }}
\def\pardiff#1#2{ \frac{\partial^2 #1 }{ \partial #2^2 }}

\def\dif#1#2{ \frac{ d #1 }{ d #2 }}

\def\diff#1#2{ \frac{ d^2 #1 }{ d #2^2 }}

\def\av#1{\langle#1\rangle}

\def\vec#1{\mbox{\boldmath $#1$} }
\def\real{\mathbb{R}}

\def\thm#1{Theorem \ref{#1}}

\def\eq#1{Eq.(\ref{#1})}

\def\fig#1{Fig. (\ref{#1})}

\def\tab#1{Table (\ref{#1})}

\def\be{\begin{equation}}
\def\ee{\end{equation}}

\def\bea{\begin{eqnarray}}
\def\eea{\end{eqnarray}}

\def\bml{\begin{multline}}
\def\eml{\end{multline}}

\def\ba{\begin{align}}
\def\ea{\end{align}}

\newenvironment{packed_enumerate}{
\begin{enumerate}[(a)]
  \setlength{\itemsep}{1pt}
  \setlength{\parskip}{0pt}
  \setlength{\parsep}{0pt}
}{\end{enumerate}}

\def\vectorthree#1#2#3{\left(\begin{array}{c}{#1}\\
                                             {#2}\\
                                             {#3}\end{array}\right)}

\def\matrixthree#1#2#3#4#5#6#7#8#9{\left(\begin{array}{ccc}{#1}&{#2}&{#3}\\
                                                           {#4}&{#5}&{#6}\\
                                                           {#7}&{#8}&{#9}\end{array}\right)}
\def\eps{\epsilon}

\def\ve{\vec{e}}

\def\vr{\vec{r}}
\def\vR{\vec{R}}

\def\vx{\vec{x}}

\def\vnabla{\vec{\nabla}}

\def\vnabla{\vec{\nabla}}

\def\ba{\overline{a}}

\newcommand{\captionfonts}{\it\small}

\makeatletter  
\long\def\@makecaption#1#2{%
  \vskip\abovecaptionskip
  \sbox\@tempboxa{{\captionfonts #1: #2}}%
  \ifdim \wd\@tempboxa >\hsize
    {\captionfonts #1: #2\par}
  \else
    \hbox to\hsize{\hfil\box\@tempboxa\hfil}%
  \fi
  \vskip\belowcaptionskip}
\makeatother   
